\documentclass[conference]{IEEEtran}
\IEEEoverridecommandlockouts

\usepackage{cite}
\usepackage{amsmath,amssymb,amsfonts}
\usepackage{algorithmic}
\usepackage{graphicx}
\usepackage{textcomp}
\usepackage{amsthm}

\usepackage{xcolor}
\def\BibTeX{{\rm B\kern-.05em{\sc i\kern-.025em b}\kern-.08em
    T\kern-.1667em\lower.7ex\hbox{E}\kern-.125emX}}

\theoremstyle{definition}

\newtheorem{theorem}{Theorem}

\newtheorem{lemma}{Lemma}

\newtheorem{definition}{Definition}

\begin{document}

\title{Two-Sided Bounds for Entropic Optimal Transport via a Rate-Distortion Integral\\
\thanks{This research 
was supported in part by NSF Grant DMS-2515510.}
}

\author{\IEEEauthorblockN{ Jingbo Liu}
\IEEEauthorblockA{\textit{Department of Statistics and Department of Electrical and Computer Engineering} \\
\textit{University of Illinois, Urbana-Champaign}\\
Urbana-Champaign, USA \\
jingbol@illinois.edu}
}

\maketitle

\begin{abstract}
We show that the maximum expected inner product between a random vector and the standard normal vector over all couplings subject to a mutual information constraint or regularization is equivalent to a truncated integral involving the rate-distortion function, up to universal multiplicative constants.
The proof is based on a lifting technique, which constructs a Gaussian process indexed by a random subset of the type class of the probability distribution involved in the information-theoretic inequality,
and then applying a form of the majorizing measure theorem.
\end{abstract}


\section{Introduction}
Optimal transport has been a pivotal tool in many fields including mathematical analysis, machine learning, and statistics.
The case where the probability measure on one side is Gaussian is particularly interesting, with applications in generative models.
Recently, \cite{liu2025simple} established the following result termed ``rate-distortion integral'': 
Given standard Gaussian measure $\gamma$ and an arbitrary probability measure $\mu$ with finite second moments on $\mathbb{R}^n$, we have
\begin{align}
\sup_{P_{YZ}\in \Pi(\gamma,\mu)} 
\mathbb{E}[\left<Y,Z\right>]
\asymp \int_0^{\infty}\sqrt{R_{\mu}(\sigma^2)}d\sigma
\label{e_rd}
\end{align}
where $R_{\mu}(\sigma^2):=\inf_{P_{UZ}\colon P_Z=\mu, \mathbb{E}[\|Z-U\|_2^2]\le \sigma^2} I(U;Z)$ denotes the rate-distortion function, $\Pi(\gamma,\mu)$ denotes the set of couplings of $\gamma$ and $\mu$, and $\asymp$ denotes equivalence up to universal constants.
We note that \eqref{e_rd} resembles the standard Dudley integral \cite{dudley2016vn}\cite{talagrand2014upper}:
\begin{align}
\mathbb{E}[\max_{t\in T}\left<Y,t\right>]
\lesssim 
\int_0^{\infty}
\sqrt{\ln N(T,\lambda)}d\lambda
\label{e_dudley}
\end{align}
where $T$ is a compact set in $\mathbb{R}^n$ and $N(T,\lambda)$ denotes its covering number at scale $\lambda$.
However, \eqref{e_rd} is a \emph{two-sided bound}, 
whereas $\lesssim$ indicates that \eqref{e_dudley} is only an upper bound up to universal constant.
The sharpness of \eqref{e_rd} can be understood via the construction in \cite{liu2025simple}, which constructs a Gaussian process indexed by the type class of $\mu$ in a certain high-dimensional space.
Due to the permutation invariance of the type class, 
this process is \emph{stationary} \cite{van2014probability}, so that Dudley's integral is well-known to be sharp \cite{van2014probability}.
Further, \cite{liu2025simple}  showed that \eqref{e_rd} implies Talagrand's celebrated majorizing measure theorem by essentially applying the data processing inequality,
which is arguably simpler than existing proofs.
Moreover, \eqref{e_rd} has other interesting consequences, 
including sharper rates in statistical regression \cite{liu2025simple} and sub-Gaussian comparison theorems 
\cite{liu2025simple}\cite{van2025subgaussian}.

In this paper, we establish 
two-sided bounds on the entropic constrained/regularized optimal transport, analogous to \eqref{e_rd} but with a constraint/regularization of mutual information.
It turns out that a two-sided inequality still holds, provided that the integral on the right side is truncated (see Section~\ref{sec_main} for precise statements of the result).
The method is similar to the lifting in \cite{liu2025simple}, but with an important refinement that replaces the type class with a randomly selected subset, 
to avoid ``overfitting'' and hence control the mutual information.
While such a random subset is no longer permutation invariant, 
the resulting process is still close to stationary, since the numbers of randomly selected points in each ball approximately the same due to a sharp double-exponential concentration (truncation of the integral precisely restricts to the regime where such concentration holds).

Regularization plays a prominent role
in designing efficient algorithms with provable convergence (see \cite{peyre2019computational} and many references therein),
and entropic regularization is the most popular choice which is used in the Sinkhorn's algorithm.
It is thus interesting to explore potential applications of the new inequalities in machine learning and optimization.
Furthermore, since our Theorem~\ref{thm2} offers a clean, ``incremental form'' of the rate-distortion integral satisfying exact tensorization, 
a tantalizing question is whether there is purely analytic proof of the majorizing measure theorem even simpler than the lifting approach in \cite{liu2025simple}.

\emph{Related work.} The lifting approach in \cite{liu2025simple} was inspired by methods used for Cover's problem \cite{wu2018capacity}.
In \cite{bai2023information}, it was shown that the Gaussian-specific bound in \cite{wu2018capacity} can be recovered using optimal transport methods without lifting.
In \cite{liu2020minoration}, another method based on convex geometry solved
Cover's problem in its original form concerning general discrete memoryless channels,
and \cite{liu2023soft} applied the lifting argument to the method of \cite{liu2020minoration}, extracting general bounds on entropic optimal transport that reduces to the one in \cite{bai2023information} in the Gaussian case.
These bounds are dimension-dependent, differing from the present paper.
More recently, \cite{chu2025talagrand} obtained upper and lower bounds on the related free-energy functional in the case of a finite index set.
Their bounds are one-sided in the general case.

\section{Preliminaries}
\emph{Notation.}
 $\wedge$ and $\vee$ to denote the min and max of two numbers.
Given probability measures $P$ and $Q$ on a metric space $(T,d)$, 
$W_p(P,Q):=\inf_{P_{XY}\in\Pi(P,Q)}\mathbb{E}^{\frac1{p}}[d(X,Y)^p]$ denotes the Wasserstein distance, where $\Pi(P,Q)$ denotes the set of couplings.
Product measures are denoted by $P\times Q$ or $P^{\otimes N}$.
$\lesssim$ means less than or equal to up to a universal constant,
and $\asymp$ means both $\lesssim$ and $\gtrsim$.
The empirical distribution of a sequence $x=(x(1),x(2),\dots,x(d))$ is denoted by $\widehat{P}_x$.
The cardinality of a set $T$ is denoted $|T|$.
Information-theoretic notation mostly follows \cite{thomas2006elements},
and logarithms and exponentials have the natural base by default.

\subsection{Information-Constrained/Regularized Optimal Transport}
Given probability measures $\gamma$ and $\mu$ on $\mathbb{R}^n$, and $R\ge 0$, define 
\begin{align}
{\sf w}(\gamma,\mu,R):=
\sup_{P_{YZ}\in \Pi(\gamma,\mu)\colon I(Y;Z)\le R} 
\mathbb{E}[\left<Y,Z\right>],
\end{align}
which we abbreviate as ${\sf w}(R)$ if no confusion.
It is related to the information-constrained optimal transport distance since $2{\sf w}(\gamma,\mu,R)=
\mathbb{E}[\|Y\|_2^2]+\mathbb{E}[\|Z\|_2^2]
-\inf_{P_{YZ}\in \Pi(\gamma,\mu)\colon I(Y;Z)\le R} 
\mathbb{E}[\|Y-Z\|_2^2]$.
Besides the constrained version, we can also define a regularized version:
For any $\beta\ge 0$, define 
\begin{align}
{\sf f}(\gamma,\mu,\beta):=\sup_{R>0}\{{\sf w}(\gamma,\mu, R)
-\beta R\}.
\end{align}
Again we may use the abbreviation ${\sf f}(\beta)$ if no confusion. We note that ${\sf f}(\beta)\ge 0$.

\subsection{Upper Bounds on Stochastic Processes}
Following \cite{talagrand2014upper},
a general stochastic process on a metric space $T$ is denoted by $(X_t)_{t\in T}$,
and a central object is $\mathbb{E}[\sup_{t\in T}X_t]$.
To avoid measurability issues in defining the supremum, \cite{talagrand2014upper} actually defined it through the \emph{lattice supremum}, i.e., $\sup_{T'\subseteq T;\, |T'|<\infty}\mathbb{E}[\max_{t\in T'}X_t]$, 
and noted that for the purpose of deriving concrete bounds there is no loss of generality in focusing on finite $T$.
If $|T|<\infty$, we can always represent $(X_t)_{t\in T}$ in an inner product form:
for example, set $T=\{e_1,\dots,e_n\}$ as the standard basis vectors, so that $X_t=\left<Y,t\right>$,
where $Y=(X_{e_1},\dots, X_{e_n})$ is a random vector.

Upper bounds on $\mathbb{E}[\sup_{t\in T}X_t]$ based on chaining-type arguments usually only require controls on the tail probabilities, such as sub-Gaussianity.
Lower bounds, such as the one in the majorizing measure theorem, are much more delicate, usually stated for Gaussian processes.
Since the current paper focuses on two-sided bounds, we will state the results for Gaussian processes, although the upper bound direction only requires sub-Gaussianity (see \cite{liu2025simple} for the treatment in the $R=\infty$ case).
For Gaussian processes with $|T|=n<\infty$, we can use the inner product representation 
\begin{align}
X_t=\left<Y,t\right>,\quad Y\sim \mathcal{N}(0, I_n).
\end{align}
Then the metric on $T$, defined as $d(t_1,t_2):=\mathbb{E}^{1/2}[|X_{t_1}-X_{t_2}|^2]$, is equivalent to $\|t_1-t_2\|_2$ on $\mathbb{R}^n$.

A Gaussian process is called \emph{stationary} if the metric on the index set is invariant under the action of a transitive group (see e.g.~\cite{van2014probability,liu2025simple}).
Under the stationarity assumption, the chaining proof of Dudley's integral \eqref{e_dudley} can be reversed, 
and it is known since Fernique \cite{fernique1975regularite} that Dudley's integral is sharp in the stationary case.

For general nonstationary Gaussian processes, 
Fernique \cite{fernique1975regularite} conjectured a form of the majorizing measure theorem,
whose nontrivial lower bound part was proved more than a decade later by Talagrand \cite{talagrand1987regularity} (see also discussions in \cite{maurey}).
Given a probability distribution $\mu$ on $T$, define 
\begin{align}
I_{\mu}(t):=\int_0^{\infty}\sqrt{\ln\frac1{\mu(B(t,\lambda))}}d\lambda,
\label{e_imu}
\end{align}
where $B(t,\lambda)$ denotes the ball centered at $t$ with radius $\lambda$, under the metric $d$. 
Talagrand's celebrated majorizing measure theorem states that $\mathbb{E}[\sup_{t\in T}X_t]$ equals, up to universal multiplicative factors,
\begin{align}
\gamma_2(T):=\inf_{\mu}\sup_{t\in T}I_{\mu}(t).
\label{e2}
\end{align}
It is also known (see \cite{talagrand2014upper}) that 
the following quantity is equivalent to $\gamma_2(T)$ up to universal constants:
\begin{align}
\delta_2(T)&:=
\sup_{\mu}\inf_{t\in T}I_{\mu}(t).
\label{e1.4}
\end{align}
While Talagrand's proofs \cite{talagrand1987regularity,talagrand1992simple} are nontrivial, the lifting method shows that the majorizing measure theorem and its equivalent forms can be derived from  Fernique's earlier results for stationary processes \cite{liu2025simple}.

\subsection{Method of Types and Rate-Distortion Theory}
The lifting argument relies on the \emph{method of types}, a standard method in information theory \cite{CsiszarKorner1981} and large deviation analysis \cite{dembo2009large}.
The \emph{type} of a sequence is defined as its empirical distribution. 
Since not all distributions can be a type for a given length of the sequence, we need the following: 
\begin{definition}\label{def2}
Suppose that $\mu$ is a distribution on a finite set $\mathcal{Z}$. We say $\mu$ is \emph{rational} if there exists integer $N>0$ such that $N\mu$ is \emph{integer}, i.e., $N\mu(z)\in\mathbb{Z}$ for any $z\in \mathcal{Z}$. 
\end{definition}
If $N\nu$ is integer, the \emph{type class of $\nu$} is defined as the set of sequences with type $\nu$.
The following result can be found in
\cite[Chapter~2, Problem~3]{CsiszarKorner1981}:

\begin{lemma}\label{lem_t2}
Suppose that $\mathcal{X}$ and $\mathcal{Y}$ are finite sets, 
$P_{XY}$ is a distribution on $\mathcal{X}\times \mathcal{Y}$, 
and $NP_{XY}$ is integer.
Let $x^N$ be a sequence of type $P_X$,
and let $\mathcal{C}\subseteq \mathcal{Y}^N$ be the type class of $P_Y$.
Define
\begin{align}
\mathcal{C}(x^N):=\{y^N\colon (x^N,y^N) \textrm{ has type $P_{XY}$}\}.
\end{align}
Then we have
\begin{align}
&\quad (N+1)^{-|\mathcal{X}||\mathcal{Y}|}
\exp(-NI(X;Y))
\nonumber\\
&\le 
\frac{|\mathcal{C}(x^N)|}{|\mathcal{C}|}
\le 
(N+1)^{|\mathcal{Y}|}\exp(-NI(X;Y)).
\end{align}
\end{lemma}

Next, we recall the rate-distortion function in information theory.
Given a probability measure $\mu$ on $T$ and $\sigma>0$, define 
\begin{align}
\Pi_{\sigma}(\mu)
:=\{P_{Z\hat{Z}}\colon P_{Z}=P_{\hat{Z}}=\mu;\,
\mathbb{E}[d^2(Z,\hat{Z})]\le \sigma^2
\}.
\end{align} 
\begin{definition}
\label{def_rd}
Given $\mu$ on a metric space $(T,d)$, define the rate-distortion function 
\begin{align}
R_{\mu}(\sigma^2):=\inf_{P_{UZ}}I(U;Z),
\end{align}
where the infimum is over probability measure $P_{UZ}$ on $T\times T$ satisfying $P_Z=\mu$ and $\mathbb{E}[d^2(U,Z)]\le \sigma^2$.
Define 
\begin{align}
r_{\mu}(\sigma):=R_{\mu}(\sigma^2),
\end{align}
and 
\begin{align}
i_{\mu}(\sigma):=\inf_{
\substack{P_{Z\hat{Z}}\in\Pi_{\sigma}(\mu)}
}
I(Z;\hat{Z}).
\end{align}
\end{definition}
Observe that
\begin{align}
r_{\mu}(\sigma)
\le 
i_{\mu}(\sigma)
\le 
r_{\mu}(\sigma/2).
\label{e21}
\end{align}
Indeed, the first inequality is obvious from the definition (relaxing the $P_U=\mu$ constraint in the infimum). 
The second holds since given $\mathbb{E}[d^2(U,Z)]\le \sigma^2/4$, we can construct $\hat{Z}$ such that $\hat{Z}$ and $Z$ are conditionally i.i.d.\ given $U$. Then $\mathbb{E}[d^2(\hat{Z},Z)]
\le \mathbb{E}[(d(\hat{Z},U)+d(U,Z))^2]
\le \sigma^2$ and $I(Z;\hat{Z})\le I(U;Z)$ by the data-processing inequality.
Due to \eqref{e21}, all our results can be stated equivalently in terms of either $i_{\mu}$, $r_{\mu}$, or $R_{\mu}$ (up to universal multiplicative constants).

\section{Main Results}\label{sec_main}
\begin{theorem}\label{thm1}
Suppose that $\gamma=\mathcal{N}(0,I_n)$, and $\mu$ is a probability measure on $\mathbb{R}^n$ with finite second moments.
Then we have 
\begin{align}
{\sf w}(\gamma,\mu, R)
\ge 
c\int_0^{\infty} \sqrt{R\wedge i_{\mu}(\sigma)}d \sigma
\label{e16}
\end{align}
where $c>0$ is any constant such that $\mathbb{E}[\max_{t\in T}X_t]\ge c\delta_2(T)$ for all Gaussian processes $(X_t)_{t\in T}$ over a finite metric space $T$ (see \eqref{e1.4}).
\end{theorem}
In \cite[Theorem~2]{liu2025simple}, it is also shown that
\begin{align}
{\sf w}(\gamma,\mu, R)
\le 
2K\int_0^{\infty} \sqrt{R\wedge i_{\mu}(\sigma)}d \sigma,
\label{e15}
\end{align}
where $K>0$ is any constant for which Dudley's integral holds. 
To be precise, \cite[Theorem~2]{liu2025simple} is stated for a countable metric space $T$,
but the extension to $T=\mathbb{R}^n$ is immediate from a limiting argument, since both sides of \eqref{e15} are continuous in $\mu$ with respect to the 2-Wasserstein metric.
Altogether, we conclude that 
\begin{align}
{\sf w}(\gamma,\mu, R)
\asymp
\int_0^{\infty} \sqrt{R\wedge i_{\mu}(\sigma)}d \sigma.
\end{align}

\begin{theorem}\label{thm2}
Suppose that $\gamma$, $\mu$, $c$, $K$ are the same as in Theorem~\ref{thm1} and \eqref{e15}. 
Define $\phi(\mu,\alpha):=\inf_{\sigma>0}\{i_{\mu}(\sigma^2)+\frac{\sigma^2}{\alpha^2}\}$, which is a decreasing function of $\alpha>0$.
We have 
\begin{align}
{\sf f}(\gamma,\mu,\beta)&\ge \frac{c}{4}\int_{\frac{4\beta}{c}}^{\infty}\phi(\mu,\alpha)d\alpha;
\\
{\sf f}(\gamma,\mu,\beta)&\le K\int_{\frac{\beta}{K}}^{\infty}\phi(\mu,\alpha)d\alpha.
\end{align}
\end{theorem}
We will see through the proof 
that
Theorem~\ref{thm1} and Theorem~\ref{thm2} are equivalent statements up to constant factors. 
However, Theorem~\ref{thm2} further enjoys the exact tensorization property (i.e., additivity under product measures): 
for standard Gaussian measures $\gamma_1$ and $\gamma_2$ and arbitrary $\mu_1$ and $\mu_2$ with bounded second moments and compatible dimensions, we have 
${\sf f}(\gamma_1\times \gamma_2,\mu_1\times \mu_2,\beta)
=
{\sf f}(\gamma_1,\mu_1,\beta)+{\sf f}(\gamma_2,\mu_2,\beta)$, 
and $\phi(\mu_1\times \mu_2,\alpha)=\phi(\mu_1,\alpha)+\phi(\mu_2,\alpha)$.

As an example, suppose that $\mu=\gamma=\mathcal{N}(0,1)$.
Then we can verify that both ${\sf f}(\gamma,\mu,\beta)$ and $ \int_{\beta}^{\infty}\phi(\mu,\alpha)d\alpha$ converge to absolute constants as $\beta\downarrow 0$, and both are order $\frac1{\beta}$ as $\beta\to\infty$.
This verifies the bounds in Theorem~\ref{thm2} up to constants.
By rescaling and change of variables, 
the statement for $\mu=\mathcal{N}(0,\sigma^2)$ is verified for general $\sigma>0$.
In turn, by tensorization, the statement is verified for $\mu=\mathcal{N}(0,\Sigma)$ for general positive semidefinite $\Sigma$.

\section{Proof of Main Results}\label{sec_proofs}

\subsection{Proof of Theorem~\ref{thm1}}
We first establish a general result relating information-constrained optimal transport (not necessarily with a Gaussian measure) to the supremum of a process over a random subset of the type class:

\begin{lemma}\label{lem1}
Suppose that $\gamma$ and $\mu$ are rational distributions on $\mathbb{R}^n$.
Define 
\begin{align}
\mathcal{I}:=\{N\in\mathbb{Z}\colon \textrm{$N>0$, $N\gamma$  and $N\mu$ are both integer}\}.
\end{align}
For each $N\in\mathcal{I}$, let $\mathcal{C}_{\mu}$ denote the type class of $\mu$, consisting of all sequences of type $\mu=P_Z$.
Select 
\begin{align}
L:=\lfloor\exp(NR)\rfloor
\end{align} 
sequences from $\mathcal{C}_{\mu}$ uniformly at random with replacement, and call this random set $\mathcal{A}$. 
Let $Y^N$ be uniform on the type class of $\gamma$ and independent of the choice of $\mathcal{A}$.
Then 
\begin{align}
{\sf w}(\gamma,\mu, R)
=
\liminf_{\mathcal{I}\ni N\to\infty}\mathbb{E}\left[\max_{z^N\in\mathcal{A}}\frac1{N}\left<Y^N,z^N\right>\right].
\label{e_66}
\end{align}
\end{lemma}

\begin{proof}
Let $I_{\rm max}:=\max_{P_{YZ}\in \Pi(\gamma,\mu)}I(Y;Z)$,
and suppose that $P_{YZ}^*$ achieves the maximum.
By considering taking the  convex combination with $P_{YZ}^*$, it is easy to see that ${\sf w}(\gamma,\mu,r)$ (abbreviated ${\sf w}(r)$) is a strictly increasing function for $r\in [0,I_{\max}]$.
For any $r\in [0, I_{\rm max}]$, 
define 
\begin{align}
\mathcal{S}(r)&:=
\{
Q_{YZ}\in \Pi(\gamma,\mu)\colon
\mathbb{E}_{Q_{YZ}}[\left<Y,Z\right>]\ge {\sf w}(r)
\};
\\
\mathcal{S}_N(r)&:=
\{
Q_{YZ}\in \mathcal{S}(r)\colon
\textrm{$NQ_{YZ}$ is integer}
\}.
\end{align}
Then 
\begin{align}
\min_{Q_{YZ}\in \mathcal{S}(r)}I_Q(Y;Z)=r
\end{align}
where the subscript in $I_Q$ indicates that the mutual information is computed under $Q_{YZ}$.
Lemma~\ref{lem:cycle-rounding} in turn shows that
\begin{align}
\lim_{\mathcal{I}\ni N\to\infty}\min_{Q_{YZ}\in\mathcal{S}_N(r)}
I_Q(Y;Z)=r.
\label{e26}
\end{align}
Indeed, the $\ge$ part in \eqref{e26} is obvious, and the $\le$ follows from the continuity of the mutual information (given the finite support) and approximation of an arbitrary coupling by type in Lemma~\ref{lem:cycle-rounding}.
Now 
by definition we see
\begin{align}
&\quad\mathbb{P}\left[\max_{z^N\in\mathcal{A}}\frac1{N}\left<Y^N,z^N\right>
\ge {\sf w}(r)\right]
\nonumber\\
&=
\mathbb{P}[\exists z^N\in \mathcal{A}\colon
\widehat{P}_{Y^Nz^N}\in
\mathcal{S}_N(r)
].
\label{e7}
\end{align}
When $r> R$ we have, from Lemma~\ref{lem_t2},
\begin{align}
\eqref{e7}
&\le \sum_{Q\in \mathcal{S}_N(r)}
\mathbb{P}[\exists z^N\in \mathcal{A}\colon
\widehat{P}_{Y^Nz^N}=Q
]
\\
&\le \sum_{Q\in \mathcal{S}_N(r)}
\left[1-
(1-(N+1)^{n}\exp(-NI_Q(Y;Z)))^L
\right]
\\
&\le 
|\mathcal{S}_N(r)|[1-(1-(N+1)^{n}\exp(-Nr))^L]
\\
&\le \exp(-Nc)
\end{align}
for all sufficiently large $N\in \mathcal{I}$, where  $c>0$ is some constant independent of $N$,
and we used the fact that the number of types is sub-exponential in the last step.
Then since $|\max_{z^N\in\mathcal{A}}\frac1{N}\left<y^N,z^N\right>|\le C$ for $C:=\max_{y\in\mathcal{Y},z\in\mathcal{Z}}|\left<y,z\right>|$ independent of $N$, we have
\begin{align}
&\quad \mathbb{E}\left[\max_{z^N\in\mathcal{A}}\frac1{N}\left<Y^N,z^N\right>\cdot 1\left\{\max_{z^N\in\mathcal{A}}\frac1{N}\left<Y^N,z^N\right>
\ge {\sf w}(r)\right\}
\right]
\nonumber\\
&\le 
C\sqrt{
\mathbb{P}\left[\max_{z^N\in\mathcal{A}}\frac1{N}\left<Y^N,z^N\right>
\ge {\sf w}(r)\right]
}
\end{align}
which vanishes as $\mathcal{I}\ni N\to\infty$.
This implies that 
\begin{align}
\lim_{\mathcal{I}\ni N\to\infty}
\mathbb{E}\left[\max_{z^N\in\mathcal{A}}\frac1{N}\left<y^N,z^N\right>
\right]
\le 
\lim_{r\downarrow R}{\sf w}(r)
={\sf w}(R).
\end{align}
Similarly, when $r< R$ we have 
\begin{align}
\eqref{e7}
&\ge \max_{Q\in \mathcal{S}_N(r)}
\mathbb{P}[\exists z^N\in \mathcal{A}\colon
\widehat{P}_{Y^Nz^N}=Q
]
\\
&=\max_{Q\in \mathcal{S}_N(r)}
\left[1-
(1-(N+1)^{n^2}\exp(-NI_Q(Y;Z)))^L
\right]
\\
&\ge 
1-\exp(-\exp(cN(R-r)))
\end{align}
for all sufficiently large $N\in \mathcal{I}$, where  $c>0$ is some constant independent of $N$.
Then \begin{align}
&\quad \mathbb{E}\left[\max_{z^N\in\mathcal{A}}\frac1{N}\left<Y^N,z^N\right>\cdot 1\left\{\max_{z^N\in\mathcal{A}}\frac1{N}\left<Y^N,z^N\right>
< {\sf w}(r)\right\}
\right]
\nonumber\\
&\ge 
-C\sqrt{
\mathbb{P}\left[\max_{z^N\in\mathcal{A}}\frac1{N}\left<Y^N,z^N\right>
< {\sf w}(r)\right]
}
\end{align}
which vanishes as $\mathcal{I}\ni N\to\infty$.
This implies that 
\begin{align}
\lim_{\mathcal{I}\ni N\to\infty}
\mathbb{E}\left[\max_{z^N\in\mathcal{A}}\frac1{N}\left<y^N,z^N\right>
\right]
\ge 
\lim_{r\uparrow R}{\sf w}(r)
={\sf w}(R).
\end{align}
\end{proof}

\begin{proof}[Proof of Theorem~\ref{thm1}]
First, we observe that it suffices to prove the theorem for the case of rational $\mu$.
This is because when $\mu$ has finite second moments, it can be approximated arbitrarily well in the 2-Wasserstein distance by rational distributions, 
and it can be shown (using Markov coupling and applying Cauchy-Schwarz and the data processing inequality) that both sides of \eqref{e16} are continuous in the 2-Wasserstein distance.
Henceforth $\mu$ is assumed to be rational (hence must have a finite support).

Suppose that $\bar{\gamma}$ is a rational distribution such that $W_2(\bar{\gamma},\gamma)\le \epsilon$.
If $Q_{\bar{Y}Z}$ is a distribution achieving the max in the definition of ${\sf w}(\bar{\gamma},\mu, R)$, 
and $Q_{Y\bar{Y}}$ is a distribution achieving the max in the definition of $W_2(\gamma,\bar{\gamma})$,
then we can construct $P_{Y\bar{Y}Z}$ under which $(Y,\bar{Y})\sim Q_{Y\bar{Y}}$,
$(\bar{Y},Z)\sim Q_{\bar{Y}Z}$,
and $Y-\bar{Y}-Z$ is a Markov chain.
We have 
\begin{align}
{\sf w}(\gamma,\mu, R)
&\ge \mathbb{E}_P[\left<Y,Z\right>]
\\
&\ge \mathbb{E}_P[\left<\bar{Y},Z\right>]
-\sqrt{\mathbb{E}[\|\bar{Y}-Y\|_2^2]
\mathbb{E}[\|Z\|_2^2]}
\\
&= 
{\sf w}(\bar{\gamma},\mu, R)
-\epsilon\sqrt{\mathbb{E}[\|Z\|_2^2]}.
\end{align}
Define 
$
\mathcal{I}:=\{N\in\mathbb{Z}\colon \textrm{$N>0$, $N\bar{\gamma}$  and $N\mu$ are both integer}\}
$.
In Lemma~\ref{lem1}, we show that 
\begin{align}
{\sf w}(\bar{\gamma},\mu, R)
\ge 
\liminf_{\mathcal{I}\ni N\to\infty}\mathbb{E}\left[\max_{z^N\in\mathcal{A}}\frac1{N}\left<\bar{Y}^N,z^N\right>\right].
\end{align}
By \cite[Lemma~3]{liu2025simple}, we can also construct coupling of $\bar{Y}^N$ uniform on the type class of $\bar{\gamma}$ and $Y^N\sim \gamma^{\otimes N}$ such that $\lim_{\mathcal{I}\ni N\to\infty}\frac1{N}\mathbb{E}[\|\bar{Y}^N-Y^N\|_2^2]
=
\lim_{\mathcal{I}\ni N\to\infty}
\mathbb{E}[W_2^2(\bar{\gamma},\widehat{P}_{Y^N})]\le 
\epsilon^2$, so that 
\begin{align}
&\quad \liminf_{\mathcal{I}\ni N\to\infty}\mathbb{E}\left[\max_{z^N\in\mathcal{A}}\frac1{N}\left<\bar{Y}^N,z^N\right>\right]
\nonumber\\
&\ge 
\liminf_{\mathcal{I}\ni N\to\infty}\mathbb{E}\left[\max_{z^N\in\mathcal{A}}\frac1{N}\left<Y^N,z^N\right>\right]
\nonumber\\
&\quad+
\liminf_{\mathcal{I}\ni N\to\infty}\mathbb{E}\left[\max_{z^N\in\mathcal{A}}\frac1{N}\left<\bar{Y}^N-Y^N,z^N\right>\right]
\label{e8}
\\
&\ge 
\liminf_{\mathcal{I}\ni N\to\infty}\mathbb{E}\left[\max_{z^N\in\mathcal{A}}\frac1{N}\left<Y^N,z^N\right>\right]
-C\epsilon 
\end{align}
where we applied the Cauchy-Schwarz inequality to the second term in \eqref{e8}, defining $C:= \max_{z\in\mathcal{Z}}\|z\|_2$.

Next we use the fact that $\delta_2$ defined in \eqref{e1.4} is equivalent to $\gamma_2$ up to a universal constant, and then specialize $\mu$ in \eqref{e1.4} to the uniform distribution $\mu_N$ on $\mathcal{A}$.
Then 
\begin{align}
\delta_2(\mathcal{A})
&\ge \min_{t\in \mathcal{A}}\int_0^{\infty}\sqrt{\ln\frac1{\mu_N(B(t,\lambda))}}d\lambda
\\
&\ge \int_0^{\infty}\min_{t\in \mathcal{A}}\sqrt{\ln\frac1{\mu_N(B(t,\lambda))}}d\lambda
\\
&=\sqrt{N}\int_0^{\infty}\min_{t\in \mathcal{A}}\sqrt{\ln\frac1{\mu_N(B(t,\sqrt{N}\sigma))}}d\sigma,
\end{align}
where $B(t,\sqrt{N}\sigma)$ is defined as the intersection of $\mathcal{C}_{\mu}$ and the ball centered at $t$ with radius $\sqrt{N}\sigma$.
Then applying Fubini's theorem and Fatou's lemma, we have 
\begin{align}
&\quad \liminf_{\mathcal{I}\ni N\to\infty}\mathbb{E}[\frac1{N}\delta_2(\mathcal{A})]
\nonumber\\
&\ge
\liminf_{\mathcal{I}\ni N\to\infty}\int_0^{\infty}\mathbb{E}\left[\min_{t\in \mathcal{A}}\sqrt{\frac1{N}\ln\frac1{\mu_N(B(t,\sqrt{N}\sigma))}}\right]d\sigma
\\
&\ge \int_0^{\infty}\liminf_{\mathcal{I}\ni N\to\infty}\mathbb{E}\left[\min_{t\in \mathcal{A}}\sqrt{\frac1{N}\ln\frac1{\mu_N(B(t,\sqrt{N}\sigma))}}\right]d\sigma.
\end{align}
For any $\sigma>0$ independent of $N$ and any $t$,
\begin{align}
L\mu_N(B(t,\sqrt{N}\sigma))
\sim {\rm Binom}\left(\frac{|B(t,\sqrt{N}\sigma)|}{|\mathcal{C}_{\mu}|},L\right).
\end{align}
Now suppose that $R>i_{\mu}(\sigma)$.
Using Lemma~\ref{lem_t2}, Lemma~\ref{lem:cycle-rounding}, and the fact that the number of types is polynomial in $N$, we can show that
\begin{align}
\left|\frac1{N}\ln\frac{|B(t,\sqrt{N}\sigma)|}{|\mathcal{C}_{\mu}|}
+
i_{\mu}(\sigma)\right|
\le c_N
\end{align}
for some vanishing sequence of $c_N>0$ (possibly dependent on $\sigma$ and $n$).
Then $\frac{|B(t,\sqrt{N}\sigma)|}{|\mathcal{C}_{\mu}|}L \ge e^{Nc}$ for some $c>0$.
From the Chernoff bound (Lemma~\ref{lem_cher}), 
we have 
\begin{align}
\mathbb{P}[|\ln\mu_N(B(t,\sqrt{N}\sigma))
+Ni_{\mu}(\sigma)|>2
]
\le \exp(-e^{cN/2})
\end{align} 
for large enough $N$.
By the union bound, 
\begin{align}
&\quad\mathbb{P}[|\ln\mu_N(B(t,\sqrt{N}\sigma))
+Ni_{\mu}(\sigma)|>2,\,
\exists t\in \mathcal{\mathcal{Z}}^N
]
\nonumber\\
&\le |\mathcal{Z}|^N\exp(-e^{cN/2})
\\
&\le \exp(-e^{cN/3})
\end{align} 
for all large enough $N$.
We then obtain
\begin{align}
&\quad\liminf_{\mathcal{I}\ni N\to\infty}\mathbb{E}\left[\min_{t\in \mathcal{A}}\sqrt{\frac1{N}\ln\frac1{\mu_N(B(t,\sqrt{N}\sigma))}}\right]
\label{e22}
\\
&\ge 
\liminf_{\mathcal{I}\ni N\to\infty}
(1-\exp(-e^{cN/3}))
\sqrt{i_{\mu}(\sigma)-\frac1{N}\ln 2}
\\
&=\sqrt{i_{\mu}(\sigma)}.
\end{align}
If $\sigma$ is smaller so that $R\le i_{\mu}(\sigma)$, we see that \eqref{e22} is lower bounded by $R$, by the monotonicity of $\mu_N(B(t,\sqrt{N}\sigma))$ in $\sigma$.
The proof is completed by collecting the sequence of bounds and taking $\epsilon\to0$.
\end{proof}

\subsection{Proof of Theorem~\ref{thm2}}
\begin{proof}
We abbreviate $\phi(\mu,\alpha)$ as $\phi(\alpha)$.
It is easy to check that $\inf_{\sigma>0}[i_{\mu}(\sigma^2)\wedge R+\alpha^{-2}\sigma^2]=\phi(\alpha)\wedge R$.
Thus by Lemma~\ref{lem_equi},
\begin{align}
&\quad\sup_{R>0}\{{\sf w}(\gamma,\mu, R)
-\beta R\}
\nonumber\\
&\ge 
\sup_{R>0}\{
\frac{c}{4}\int_0^{\infty} \inf_{\sigma>0}[i_{\mu}(\sigma^2)\wedge R+\alpha^{-2}\sigma^2]d \alpha
-\beta R\}
\\
&=\sup_{R>0}\{\frac{c}{4}\int_0^{\infty}\phi(\alpha)\wedge R\, d\alpha -\beta R\}
\\
&=\sup_{R>0}\{\frac{c}{4}\int_0^R\phi^{-1}(r)dr -\beta R\},
\label{e29}
\end{align}
where $\phi^{-1}$ denotes the inverse function of $\phi$, which is also an increasing function nonnegative.
Then the supremum in \eqref{e29} is achieved when $\frac{c}{4}\phi^{-1}(R)=\beta$. Plugging $R=\phi(\frac{4\beta}{c})$ back, 
we see \eqref{e29} equals $\frac{c}{4}\left[\phi(\frac{4\beta}{c})\cdot \frac{4\beta}{c}+\int_{\frac{4\beta}{c}}^{\infty}\phi(\alpha)d\alpha\right]-\beta\phi(\frac{4\beta}{c})
=\frac{c}{4}\int_{\frac{4\beta}{c}}^{\infty}\phi(\alpha)d\alpha$.
This proves the lower bound direction.
For the upper bound, note that again by Lemma~\ref{lem_equi},
\begin{align}
&\quad \sup_{R>0}\{{\sf w}(\gamma,\mu, R)
-\beta R\}
\nonumber\\
&\le 
\sup_{R>0}\{
\frac{2K}{2}\int_0^{\infty} \inf_{\sigma>0}[i_{\mu}(\sigma^2)\wedge R+\alpha^{-2}\sigma^2]d \alpha
-\beta R\}
\end{align}
which, following steps similar to those in the lower bound proof, is seen to be equivalent to
$K\int_{\frac{\beta}{K}}^{\infty}\phi(\alpha)d\alpha$.
\end{proof}

\section{Auxiliary Lemmas}
The following result follows from the Chernoff bound, which can be found for example in \cite{vershynin2018high}:
\begin{lemma}\label{lem_cher}
Let $X_1,\dots,X_L$ be i.i.d.\ Bernoulli random variables with parameter $p$. 
Let $S_L:=\sum_{i=1}^LX_i$ and $\mu:=\mathbb{E}[S_L]$.
Then
\begin{align}
\mathbb{P}[S_L\ge t]
\le e^{-\mu}\left(\frac{e\mu}{t}\right)^t,
\quad\forall t\in [\mu,\infty);
\\
\mathbb{P}[S_L\le t]
\le e^{-\mu}\left(\frac{e\mu}{t}\right)^t,
\quad\forall t\in(0,\mu].
\end{align}
\end{lemma}

The following simple observation (see for example \cite[Lemma 7]{liu2025simple}) is crucial for several results in \cite{liu2025simple},
including deducing the majorizing measure theorem from \eqref{e_rd},
and bounding \eqref{e_rd} in the case of Gaussian vector $\mu$.
\begin{lemma}\label{lem_equi}
Suppose that $y\in(0,\infty)$ is a decreasing function of $x\in(0,\infty)$.
Then
\begin{align}
2\int_0^{\infty} ydx
\le 
\int_0^{\infty}\min_x\{\alpha^{-2}x^2+y^2\}d\alpha
\le 
4\int_0^{\infty} ydx
\end{align}
\end{lemma}

\begin{proof}
For any $\alpha\in(0,\infty)$, there is a unique $x(\alpha)$ satisfying $\alpha=\frac{x(\alpha)}{y(x(\alpha))}$. Then $x$ and $y$ can be parameterized by $\alpha\in(0,\infty)$, and 
$
\int_0^{\infty}
(\alpha^{-2}x^2+y^2)d\alpha
=
\int_{\alpha=0}^{\infty}
2y^2\cdot\frac{ydx-xdy}{y^2}
=4\int_0^{\infty}
ydx
$.
The claim then follows since for any $\alpha$, $\frac{\min_{x\in(0,\infty)}\{\alpha^{-2}x^2+y^2(x)\}}{\alpha^{-2}x^2(\alpha)+y^2(x(\alpha))}\in[1/2,1]$.
\end{proof}
A similar idea (though in the form of series summation rather than integral) was previously used in proving equivalent forms of the majorizing measure theorem \cite{talagrand1994constructions}.

As for approximating a general coupling by a type, we need the following observation; see for example \cite[Lemma~4]{liu2025simple} for a proof.
\begin{lemma}\label{lem:cycle-rounding}
Suppose that $\mathcal{X}$ and $\mathcal{Y}$ are finite sets,
$N>0$ is an integer,
and $P_{XY}$ is a distribution with the property that both $NP_X$ and $NP_Y$ are integer.
Then there exists $Q_{XY}$ satisfying $Q_X=P_X$, $Q_Y=P_Y$, $\max_{x,y}|Q_{XY}(x,y)-P_{XY}(x,y)|\le \frac1{N}$,
and that $NQ_{XY}$ is integer.
\end{lemma}

\newpage
\bibliographystyle{IEEEtran}
\bibliography{references.bib}
\end{document}